\documentclass{article}     
\usepackage{graphicx}
\usepackage{amsthm}
\usepackage{amssymb}
\usepackage{amsmath}
\usepackage{color}

\newtheorem{theorem}{Theorem}
\newtheorem{definition}[theorem]{Definition}
\newtheorem{lemma}[theorem]{Lemma}

\newtheorem{example}[theorem]{Example}

\newtheorem{proposition}[theorem]{Proposition}
\newtheorem{protocol}{Protocol}
\newtheorem{remark}{Remark}

\begin{document}

\title{Linear spaces and transversal designs: k-anonymous combinatorial configurations for anonymous database search}

\author{Klara Stokes          \and Oriol Farr\`as 
}


\maketitle

\begin{abstract}
Anonymous database search protocols allow users to query  a database anonymously. 
This can be achieved by letting the users form a peer-to-peer community and post queries on behalf of each other. 
In this article we discuss an application of combinatorial configurations (also known as regular and uniform partial linear spaces) to a protocol for anonymous database search, as defining the key-distribution within the user community that implements the protocol.   
The degree of anonymity that can be provided by the protocol is determined by properties of the neighborhoods and the closed neighborhoods of the points in the combinatorial configuration that is used. 
Combinatorial configurations with unique neighborhoods or unique closed neighborhoods are described and we show how to attack the protocol if such configurations are used. 
We apply $k$-anonymity arguments and present the combinatorial configurations with $k$-anonymous neighborhoods and with $k$-anonymous closed neighborhoods. 
The transversal designs and the linear spaces are presented as optimal configurations among the configurations with $k$-anonymous neighborhoods and $k$-anonymous closed neighborhoods, respectively. 
\end{abstract}

\section{Introduction}
\label{intro}

\emph{Anonymous database search} is the discipline 
dedicated to the study of the anonymity of query searches in 
databases.
Namely, it is dedicated to
the study of protocols that allow a user to retrieve information 
from a database server without revealing for the server who he is. 
Anonymous database search has also been called 
\emph{user-private information retrieval} (UPIR) \cite{DoBr,DoBrWuMa,StBr_Optimal}. 
Another example of anonymous database search, similar to the previously mentioned UPIR protocols, is the protocol Crowds~\cite{crowds}. 
Unlike private information retrieval (PIR) protocols, UPIR protocols
do not hide the query for the database. 

Usually, when people access to a database server, 
the responsibility of guaranteeing their privacy is 
assigned to the database owner or to a trusted third party. 
In UPIR protocols, the privacy 
of the user is put in the hand of the user.

The protocols for UPIR presented in~\cite{DoBr,DoBrWuMa,StBr_Optimal,SwansonStinson}
are defined on a P2P community, and are called P2P UPIR.
In a P2P UPIR protocol, a community of P2P users agree to
collaborate in order to search the database anonymously.
The users send queries to the database on behalf of other 
users, without revealing to the database the identity of 
the owner of the query. In this way, the query profiles
of the users are diffused among the rest of the users in the community.
Moreover, the protocol distributes the queries uniformly
to avoid tracing of the queries. 

The queries are written and read by users on some memory sectors that are called communication spaces.
Each user has access only to some of these communication spaces. 
As a mean to hinder unauthorized entities from obtaining 
information about queries, 
the information on the communication spaces is encrypted. 

If the encryption of the information on the communication spaces
uses the same key over the entire network, 
then there is a high risk that the key is compromised. 
But if the number of keys is too large, then the users can have 
storage problems. 
Given a set of requirements on the protocol, the search for the 
optimal distribution of communication spaces and keys 
can be defined as a combinatorial problem with constraints.

In the articles \cite{DoBr,DoBrWuMa,StBr_Optimal,StBr_PAIS2011}, combinatorial configurations, 
also known as regular and uniform partial linear spaces, 
were used to manage the distribution of communications spaces and keys for P2P UPIR. 
These combinatorial structures have been used for
distributing keys also in other contexts, as in \cite{Stinson1}. 
In \cite{SwansonStinson} the P2P UPIR protocols were modified and 
extended to more general families of block designs. 

Curious users should be prevented from obtaining information about other users. 
It can be proved that this prevention is simplified if we avoid, for instance, 
designs in which a pair of users share 
more than one communication space, and designs in which a user can 
access all communication spaces. 
This motivates the use of combinatorial configurations. 

In this article, we apply $k$-anonymity arguments to the construction of anonymous database search algorithms, and present the combinatorial conﬁgurations with $k$-anonymous (open) neighborhoods and with $k$-anonymous closed neighborhoods. 
We study transversal designs and linear spaces and we show that they are optimal configurations. 
We present two versions of 
P2P UPIR protocols that are based
on the protocols presented in~~\cite{DoBr,DoBrWuMa,StBr_PAIS2011,SwansonStinson}. 
In one of the protocols (P2P UPIR 1), the user cannot forward directly his own
queries to the database, while in the other he is allowed to do so (P2P UPIR 2).

The P2P UPIR 2 protocol is a modification of P2P UPIR 1, 
designed in order to avoid so-called neighborhood attacks in some combinatorial configurations. 
These attacks were first described in \cite{StBr_PAIS2011},  and are based on query repetition from the P2P UPIR users, 
in combination with unique neighborhoods in the combinatorial configuration. 
As observed in \cite{SwansonStinson}, a neighborhood attack can be modeled as the intersection of neighborhoods, 
that may return a single identified point in the case of a unique neighborhood. 
In this article we give several examples of combinatorial configurations with unique neighborhoods. 

The results presented in this article, shows that it is not necessary for users to self-submit their queries in order to avoid neighborhood attacks. 
We present a family of combinatorial configurations with anonymous neighborhoods. 
We use the concept of $k$-anonymity to measure the degree of this anonymity, 
and say that a point has a $k$-anonymous neighborhood if it is the neighborhood of at least $k-1$ other points. 
Then we study and characterize the combinatorial configurations with $k$-anonymous neighborhoods. 
In particular, we justify why the transversal designs can be regarded as optimal among the configurations with $k$-anonymous neighborhoods for P2P UPIR 1. 
We also characterize the anonymity that is provided by the P2P UPIR 1 protocol, when combinatorial configurations with $k$-anonymous neighborhoods are used. 

As can be deduced from  \cite{StBr_Optimal}, the linear spaces are optimal configurations with respect to maximal diffusion of the query profiles. 
In \cite{SwansonStinson}, the linear spaces were presented as the only combinatorial configurations that can provide P2P UPIR with so-called \emph{perfect anonymity}, and this was extended to designs in general, among which the covering designs were distinguished for having the same property. 

In this article we show that the linear spaces have $k$-anonymous closed neighborhoods, 
and that they maximize the parameter $k$. 
This justifies once again why the linear spaces are optimal for P2P UPIR. 
We also construct a new class of combinatorial configurations that also have $k$-anonymous closed neighborhoods, but in which $k$ is smaller compared to the linear spaces.

This article is structured as follows. 
Section~\ref{sec:prel} contains the preliminaries, 
in Section~\ref{sec:P2PUPIR} we define the P2P UPIR protocols,  
Section~\ref{sec:notations} introduces notation and formalizes P2P UPIR in terms of data privacy. 
In Section~\ref{sec:compprot} we describe attacks on P2P UPIR and identify ways to avoid some of them. 
Section~\ref{sec:privP2PUPIR} discusses the nature of the privacy that can be provided by the P2P UPIR protocols and under what conditions.  
In Section~\ref{sec:unique_neighborhoods} we give examples of combinatorial configurations with unique neighborhoods. 
Sections~\ref{sec:combconfnanneighbors} and \ref{sec:nanonymousclosedneighborhoods} classify and give examples of combinatorial configurations with $k$-anonymous neighborhoods and $k$-anonymous closed neighborhoods, respectively. 
The article ends with conclusions.

\section{Preliminaries}
\label{sec:prel}
This section contains known results and concepts that will be used in the rest of the article. 

\subsection{Privacy}

Most of the definitions and the notation presented in the next paragraph are taken from \cite{Privacy}.

An \emph{adversary}, or an \emph{attacker}, is an entity that aims for the destruction of the privacy protection. 
A subject $s$ is \emph{anonymous} if the adversary cannot identify him within a set of subjects. 
We call this set of subjects the \emph{anonymity set} of $s$.
Two or more so-called items of interest are \emph{unlinkable} if within the system (comprising these and possibly other items), the adversary cannot sufficiently distinguish whether these items of interest are related or not.
In this article, an item of interest can be a query, a sequence of queries, the owner of a query, the owner of a sequence of queries, or the identity of this owner. 
We can express anonymity in terms of unlinkability; anonymity is provided if it is not possible to link a subject to the identity of this subject. 
\emph{Confidentiality} is the quality of being prevented from the disclosure of information to unauthorized individuals or systems. 

Disclosure risk control for statistical databases is a research area concerned with the protection of the privacy of individuals in published statistical databases~\cite{Josep2,Willenborg}.
The naive solution is to protect the database by simply removing the identifiers (e.g. name, ID, social security number) from the tables.
It is well-known that this solution is far from satisfactory. 
In many cases it is rather easy to recover the identifier of the anonymized record (see for example~\cite{Sweeney}). 
Other more sophisticated solutions for protecting databases have been proposed. 
Many of these solutions are methods for obtaining so-called $k$-anonymity, which we will define below. 

A \emph{database} is a collection of records of data.
We may assume that all records correspond to distinct individuals or objects.  
Every record has a unique identifier and is divided into attributes. 
The attributes can be very specific, as the attributes ``height'' or ``gender'', or more general, as the attributes ``text'' or ``sequence of binary numbers''. 

Suppose that the database can be represented as a single table. 
Let the records be the rows of the table and let the attributes be the columns. 
The intersection of a row and an attribute is a cell in the table, and we call the data in the cells the entries of the database.  
Also other data structures, like for example graphs, or in general, incidence structures,  
are representable in table form. 

Let $T$ be a table with the set of attributes $A$. 
Let $B\subseteq A$ be a subset of these attributes. 
We denote the projection of the table on the attributes $B$ by $T[B]$.
We suppose that every record contains information about a unique individual. 
An \emph{identifier} $I$ in a database is an attribute such that it uniquely identifies the individuals behind the records. 
In particular, any entry in $T[I]$ is unique. 
A \emph{quasi-identifier} $QI$ in the database is a collection of attributes $\{A_1,\dots,A_n\}$, 
such that they in combination can uniquely, or almost uniquely, identify a record \cite{Dalenius}.  
That is, the structure of the table allows for the possibility that an entry in $T[QI]$ is unique, or that there are only a small number of equal entries. 
In the former case the entry in $T[QI]$ uniquely identifies the individual behind the record and in the latter, the few other individuals with the same entries in $T[QI]$ may form a collusion and use secret information about themselves in order to make this identification possible.  

\begin{example}
If a table contains information on students in a school class, the attributes birth data and gender could be sufficient to determine to which individual a record of the table corresponds, although it is possible that not all records will be uniquely identified in this way. 
Hence for this table, birth date and gender are an example of a quasi-identifier.
\end{example}
The following definition of $k$-anonymity appeared for the first time in~\cite{ref:Samarati.Sweeney.1998} (see also~\cite{Sweeney}).

\begin{definition}\label{1:def:nan}
A table $T$, that represents a database and has associated quasi-identifier $QI$, is $k$-\emph{anonymous} if every sequence in $T[QI]$ appears with at least $k$ occurrences in $T[QI]$. 
\end{definition}

\subsection{Incidence Structures}
\label{sec:incidencegeometry}
An \emph{incidence structure} $(P,L,I)$ 
sometimes also called a \emph{block design}, 
consists of a \emph{point set} $P$, a family of subsets of points $L$ called \emph{block set},  
and an \emph{incidence relation} $I$ on $P$. 
The elements of $P$ and $L$ are called \emph{points} and \emph{blocks}, respectively,  and two points are related by the incidence relation if and only if there is a block containing both. In this article we assume that the same block only can appear once in $L$. 
We will also assume that the incidence structures are connected, 
so that for every two different points $p,q\in P$, 
there is a chain of incidences starting with $p$ and ending with $q$. 

If all blocks have the same number of points $k$, 
then we say that the incidence structure is \emph{$k$-uniform}.
If all points are in the same number of blocks $r$, 
then we say that the incidence structure is \emph{$r$-regular}. 
The \emph{order} of a uniform and regular incidence structure is the integer pair $(k-1,r-1)$.
 
A \emph{parallel class} $L'$ is a subset of $L$ such that for all $p\in P$ there is a unique block $l\in L'$ such that $p\in l$. Every parallel class $L'=\{l_1,\ldots,l_m\}$ is a partition of $P$, because 
$P=\cup_{i=1}^{m} l_i$ and $l_i\cap l_j=\emptyset$ for $1\leq i< j\leq m$. An incidence structure $(P,L,I)$ is \emph{resolvable} if there exists a partition of the set of blocks $L=\{L_1,\dots,L_s\}$,
such that $L_i$  is a parallel class of blocks for $1\leq i\leq s$. 

The \emph{line} spanned by two points is the intersection of the all blocks containing these points. 
When every pair of points is contained in at most one block, then the blocks are the lines. 

In an incidence structure in which the blocks are lines, we say that two points are \emph{collinear} if there is a line through the two points. 
Observe that a point is always collinear with itself. 
We define the \emph{closed neighborhood} $CN(p)$ of $p$  as the set of points that are collinear with the point $p$. 
The \emph{neighborhood} or \emph{set of neighbors} $N(p)$ of a point $p$ is the set of points in $C$ that are collinear with $p$ but different from $p$.  
Observe that $CN(p)=N(p)\cup\{p\}$. 

\subsubsection{Linear Spaces, Partial Linear Spaces, and Combinatorial Configurations}

A \emph{linear space} is an incidence structure in which every two points are in exactly one block,  
so in a linear space we may say that the blocks are lines.  
It is not required for the lines to have the same number of points, but the minimum number of points on a line is two. 

A \emph{finite affine plane} is an incidence structure in which  
\begin{itemize}
\item every two points span exactly one line, 
\item for every point $p$ and line $l$ not incident with $p$, there is exactly one other line $m\in L$ such that $p$ is incident with $m$ and $l\cap m=\emptyset$,
\item there is a triangle, a set of three points such that they pairwise span different lines. 
\end{itemize}
A finite affine plane is therefore a linear space. 
In a finite affine plane there is always a natural number $n$ such that there are $n$ points on every line, $n+1$ lines through every point, $n^2$ points and $n^2+n$ lines. 
In particular, a finite affine plane is uniform and regular. 
The order of a finite affine plane is $(n-1,n)$, but the usual notation is that the order is $n$.
The second condition in the definition of affine plane implies that the set of lines is partitioned into classes of parallel lines, so that an affine plane is resolvable. 


The affine plane over a finite field of order $q$ is always a finite affine plane of order $q$ which we denote by $\mathbb{A}(\mathbb{F}_q)$. 
From this it is deduced that there exists a finite affine plane of order $q$ for every prime power $q$. 
When $q$ is not prime, then there are other finite affine planes than $\mathbb{A}(\mathbb{F}_q)$, 
but it is not known if there are finite affine planes of order $n$ when $n$ is not a power of a prime. 
It is conjectured that there exists a finite affine plane of order $n$ if and only if $n$ is a power of a prime. 
 
Other examples of uniform and regular finite linear spaces are the finite projective planes, the unitals, and the Denniston designs. 

A \emph{partial linear space} is an incidence structure in which every two points can be on at most one line.
Also it is required that the minimum number of points on a line is two. 
All linear spaces are partial linear spaces. 
The number of points is usually denoted by $v$, and the number of lines by $b$.
In this article we will concentrate on $r$-regular and $k$-uniform partial linear spaces, also known as \emph{combinatorial $(v,b,r,k)$-configurations, } or shorter, $(r,k)$-configurations.  
%
%
For general references on combinatorial configurations, see \cite{Gropp,Grunbaum}. 

\subsubsection{$t$-Designs and Transversal Designs}
\label{sec:transversaldesigns}

Another interesting type of incidence structure are the $t$-designs. A $t$-\emph{design} with parameters $(v,k,\lambda)$ has $v$ points, $k$ points in every block and every $t$-element subset of points appear in exactly $\lambda$ blocks. 
A $2$-design with $\lambda=1$ is a combinatorial configuration, or more precisely, a regular linear space.  
The $2$-designs are also called \emph{balanced incomplete block designs} with parameters $(v,k,\lambda)$, or shorter, \emph{$(v,k,\lambda)$-BIBD}. 

In this article we will also treat a third type of incidence structure, the transversal designs. 

A \emph{transversal design} $TD_{\lambda}(k,n)$ is a $k$-uniform incidence structure $(P,L,I)$ with $|P|=kn$ that admits a partition  of $P$ whose parts, called \emph{groups}, have cardinality $n$, and satisfy the following properties:
\begin{enumerate} 
\item any group and any block contain exactly one common point, and
\item every pair of points from distinct groups is contained in exactly $\lambda$ blocks.
\end{enumerate}

In a transversal design the set of groups forms a partition of $P$ but it is not a parallel class since the groups are not blocks. 
On the other hand, if the block set $L$ can be partitioned into parallel classes, then we get a resolvable transversal design. 

A transversal design $TD_{\lambda}(k,n)$ is a combinatorial $(kn,n^2,n,k)$-configuration if and only if $\lambda=1$. In this article we are interested in the transversal designs of this kind. 
For simplicity of notation we will denote a transversal design with $\lambda=1$ by $TD(k,n)$. 
It is well-known that affine planes can be used to construct transversal designs as described in the following lemma. 

\begin{lemma}\label{1:lem:transaffine}
Whenever there exists a finite affine plane of order $n$, then for every $2\leq k\leq n$ there exists a transversal design $T(k,n)$. 
\end{lemma}

\begin{proof}
As point set $P$ of $TD(k,n)$, take the points on the $k$ lines from one of the parallel classes of an affine plane of order $n$. 
As the groups of $TD(k,n)$, take the lines from the same parallel class. 
As lines of $TD(k,n)$, take the lines in the rest of the parallel classes of the affine plane, restricted to the points in $P$. 
\end{proof}

More generally, it is well-known that the existence of transversal designs is related to the existence of a set of mutually orthogonal latin squares. For more information about these structures, see for example \cite{MOLS,PaBoSh}. 
%


\section{P2P UPIR: peer-to-peer protocols for anonymous database search} \label{sec:P2PUPIR}

In this section we describe the peer-to-peer protocols 
for user-private information retrieval (P2P UPIR), first presented in~\cite{DoBr,DoBrWuMa}.  
These protocols use communication spaces, that are memory sectors in which 
a user who has access to the corresponding cryptographic key can write and read queries and the answers to these queries. 
The distribution of the cryptographic keys is determined by a combinatorial configuration.
The clients of the protocol are mapped to the points of a combinatorial configuration, and
the keys, or the communication spaces, are mapped to the lines.
The result is that a client, represented by the point $p$, has the cryptographic keys giving access to the communication spaces that are represented by the lines through $p$.

\subsection{The P2P UPIR INIT protocol}
The P2P UPIR protocols described herein are called by a protocol that is implemented by all the community of users together. 
We call this protocol P2P UPIR INIT.
This protocol takes as parameter the combinatorial configuration for the distribution of communication spaces. 

By abuse of notation, we will not distinguish the points and the lines of the configuration from the users and the communication spaces they represent.
A communication space is a queue of messages, together with a cryptographic key from a symmetric cipher, used to encrypt the messages.

The precondition is here that a community $P$ of $n$ users wants to implement a P2P UPIR protocol. 
The postcondition is that some user has dropped out of the protocol.
\begin{protocol}[P2P UPIR INIT]~\label{1:def:p2pupirinit}
\begin{enumerate}
\item The users in $P$ are mapped to the points of the combinatorial configuration. 
\item The users repeat execution of the P2P UPIR protocol with frequency $f$ (which is not required to be constant, nor the same for all users). 
\end{enumerate}
\end{protocol}


\begin{remark}\label{rmk:freqcheck}
The protocol described in \cite{DoBr,DoBrWuMa} was different. 
For instance, the user repeated the P2P UPIR protocol only when they had a query to post. 
However, in order to limit the waiting time before a user can post his query, and the response time for the answer, the period of protocol repetition must be bounded. More exactly, $f$ should always be higher or equal to the highest query submission frequency among the users. 
\end{remark}

\begin{remark}
It is not necessary to end the P2P UPIR INIT protocol only because a user $p$ is temporally away. 
The owner of a query that should have been posted by $p$, can execute P2P UPIR again in order to get his query posted to the server. 
However, we think that only modest and controlled absences should be allowed for, since a prolonged absence causes the deterioration of the provided anonymity.
\end{remark}

\subsection{The P2P UPIR 1 protocol}
First we present a P2P UPIR protocol which is similar to the protocol described in \cite{DoBr,DoBrWuMa}, but modified following the ideas from \cite{SwansonStinson}. 
We will call this protocol P2P UPIR 1.
The individual $p$ is member of a community of user implementing the P2P UPIR INIT protocol, and $p$'s execution of P2P UPIR 1 is done within P2P UPIR INIT.
The user $p$ may, or may not, have a query $Q$ which he wants to post to the community.

\begin{protocol}[P2P UPIR 1]~
\begin{enumerate}
\item The user (point) $p$ selects uniformly at random a communication space (line) $l$ passing through $p$;
\item $p$ decrypts the content on $l$ using the corresponding cryptographic key.
The outcome is a queue of messages $M=(M_i)$.
For every message $M_i$ in the queue:
\begin{itemize}
\item If $M_i$ is a \textbf{query addressed to $p$}, then $p$ removes $M_i$ from the queue, forwards $M_i$ to the server, receives the answer $A$, encrypts $A$ and writes $A$ to the end of the queue $M$;
\item Else if $M_i$ is an \textbf{answer to a query belonging to $p$}, then $p$ reads $M_i$ and removes $M_i$ from the queue $M$; 
\item Else, $p$ leaves $M_i$ on the queue without action;
\end{itemize}
\item If $p$ has a query $Q$, then 
\begin{enumerate}
\item  $p$ selects uniformly at random a point $p'\neq p$ on $l$; 
\item  $p$ addresses $Q$ to $p'$ and writes $Q$ to the end of the queue $M$.
\end{enumerate}
\end{enumerate}
\end{protocol}

\subsection{The P2P UPIR 2 protocol}
We also present a variation of the former protocol which we will call P2P UPIR 2. This protocol was first described in~\cite{StBr_PAIS2011}.  
 
The P2P UPIR 2 protocol differs from the P2P UPIR 1 protocol only in how the users forward their own queries. 
We say that a user who forwards his own queries with probability $x$, has self-submission $x$. 
The P2P UPIR 2 protocol with self-submission $x$ is obtained from the P2P UPIR 1 protocol by replacing step 3 (a) by: \bigskip\\
\textit{3 (a') $p$ selects a point $p'$ on $l$; with probability $x$ he selects $p'=p$, else he selects uniformly at random $p'\neq p$ on $l$;}
\bigskip
\begin{remark}
As will be proved in Proposition~\ref{prop:selfsubmission}, the P2P UPIR 2 protocol should be executed with self-submission $x=\frac{1}{|CN(p)|}=\frac{1}{r(k-1)+1}$. 
\end{remark}

\section{Notations and formal framework for the analysis of P2P UPIR}
\label{sec:notations}


In this section we will define the formal framework in which the rest of the analysis will take place. 
\subsection{Queries}

Let $P$ be a community of users implementing an instance of the P2P UPIR protocol. 
For every user $p\in P$ we define the \emph{real query profile} $RP(p)$ as the temporal sequence of queries which $p$ posts to the communication spaces and 
 the \emph{apparent query profile} $AP(p)$ as  the temporal sequence of queries which the user posts to the server.
By extension we define the real query profile $RP(U)$ and the apparent query profile $AP(U)$ of a set of users $U\subseteq P$. 


A \emph{query} is a set of one or more search terms. 
A \emph{repeated query} is a query which occurs more than once in the real profile of a user. 
A \emph{repeated variation of a query} is a query posted by a user which is a slight modification of a previous query posted by the same user. 
The latter definition is vague and ambiguous, but still useful. 

We say that a profile is \emph{rare} if it contains many unique queries or unique combinations of queries and we say that it has repetition if it contains many repeated queries or repeated variations of queries.

\subsection{Disclosure control}
\label{sec:disclosurecontrolp2pupir}

In \cite{Josep1}, three types of privacy protection are distinguished, in function of the entity for whom the protection is provided: \emph{respondent privacy}, \emph{owner (holder) privacy} and \emph{user privacy}. 
The aim of the P2P UPIR protocol is to provide anonymous database search, which falls under the area of anonymous communication, or, following the notation in \cite{Josep1}, user privacy.  
However, in this article we choose to model it in the context of respondent privacy, as a method for disclosure control of databases. 
The database protected by P2P UPIR is then the collection of queries that the users of the protocol post to the server, or in other words, the real profiles of the users. 

There are two important differences between traditional  disclosure control for statistical databases (respondent privacy) and disclosure control for P2P UPIR (user privacy interpreted as respondent privacy): 
\begin{itemize}
\item Typically, in respondent privacy, the disclosure control is applied to a given database. 
The P2P UPIR protocol is however executed in real-time as the users post queries to the server, that is, as the information is introduced into the database. 
We may say that the P2P UPIR is a \emph{streaming} disclosure control method;
\item For respondent privacy, it is typically important to balance low disclosure risk with low information loss, 
since it is useless to publish a database without information. 
In the P2P UPIR protocol, the users anonymize the data they give to the server themselves, instead of leaving this task to the server. 
For the aim of P2P UPIR, there is no need to control the utility of the query profiles collected by the server. 
We will assume that the users have no interest in providing a useful statistical database.
\end{itemize}
Some users find useful some of the services provided by the server that are based on their query (or mail) profile. 
Also, typically, the server provides query searches for free, in exchange for the valuable information that is collected in the query profiles.  
We propose that the query profile should be maintained by the user himself, and provided to the server when so desired.
This approach would put the privacy of the user in the hand of the user, where it should be.

Consider a community of $v$ users $P$ implementing an instance of the P2P UPIR 1 protocol.  
Without loss of generality, we can limit the analysis to some time interval $t$. 
Then we note by $RP_t(P)$ and $AP_t(P)$ the profiles $RP(P)$ and $AP(P)$ restricted to $t$.
In this context, the P2P UPIR protocol is a transformation of the database which we will denote by
$$\begin{array}{rccl}
\rho:&D&\rightarrow &D\\\\
&RP_t(P)&\mapsto& AP_t(P),
\end{array} $$
where $D$ is the space of all possible query databases. 
 
The database $RP_t(P)$ is a table 
where the identifier is the user ID, 
and there is an attribute Q($t_i$) for every approximate time interval $t_i$, 
containing a single query posted to the server by the user approximately at time $t_i$, or a null entry.  

After applying $\rho$ to this table we obtain the transformed database $AP_t(P)$. 
The action of $\rho$ can be described as swapping the data in the column $Q(t_i)$, 
under the constraint that the content in the record $p$ can be replaced by the content in the record $q$ only if $p\in N(q)$.
Disclosure control methods of this type are called data swapping (first appearance in~\cite{DaleniusReiss}). 

Observe that $\rho$ also adds some noise to the time stamps of the queries. 
For example, the fact that $p$ posts first $Q_1$ and then $Q_2$ to the community of users, does not imply that $Q_1$ is posted before $Q_2$ to the server. 
Therefore, in order for the swapping to preserve columns, we should think of the queries in $AP_t(P)$ as sorted according to the time they are posted to the community of users, not according to the time they are posted to the server.

\section{Attacking P2P UPIR}
\label{sec:compprot}
In this section we will discuss attacks on P2P UPIR and some countermeasures.   

The purpose with the P2P UPIR protocol is to protect the privacy of the user when retrieving information from a server. 
Therefore our main concerns are attacks from the server, or adversaries that have characteristics similar to the server. 
We will also briefly consider attacks from other users.

\subsection{Neighborhood attacks on P2P UPIR 1}
\label{sec:attack1}
In the P2P UPIR 1 protocol the user forwards to the server only queries from collinear users different from himself. 

Consider a community of users implementing the P2P UPIR 1 protocol and suppose that the initialization protocol is given a combinatorial configuration such that there are points with unique neighborhoods. 
That is, there are points $p\in P$ such that $N(q)\neq N(p)$ for every $q\in P$, $q\neq p$. 
The users are mapped to the points in the combinatorial configuration, and a user $p$ will share communication spaces with the set of users $N(p)$, 
so that the users who post the queries in $RP(p)$ are the users in $N(p)$. 
Now suppose that the user $p$ repeats the same query over and over again. 
After a while, the probability that all users in $N(p)$ have posted the query will be high. 
Therefore, since we know that $p$ is the only user with the neighborhood $N(p)$, if the query is rare, then we will be able to link the query to the user $p$, and so the anonymity provided by the protocol is broken. 

The article \cite{StBr_PAIS2011} discussed the fact that it is very common that users of web-based search engines post the same or a slightly modified version of the same query several times. 
Other references on this subject are \cite{SpWoJaSa,Teevan}.

Examples of combinatorial configurations with unique neighborhoods
are provided in Section~\ref{sec:unique_neighborhoods} and combinatorial configurations with $n$-anonymous neighborhoods are discussed in Section~\ref{sec:combconfnanneighbors}.

\subsection{Adjusting query self-submission for P2P UPIR 2}

\label{sec:modification}

We just saw that the P2P UPIR 1 protocol, 
which is similar to the version of the P2P UPIR protocol that appears in \cite{DoBr,DoBrWuMa}, 
can be attacked if the configuration that is used has points with unique neighborhoods. 
Examples of combinatorial configurations with unique neighborhoods are the linear spaces (see Section~\ref{sec:unique_neighborhoods}). 

This is a problem, since otherwise the linear spaces are optimal configurations for P2P UPIR, 
if we consider the anonymity of the user in front of the server.
The use of a linear space maximizes the number of apparent profiles into which the real profile of a user is diffused, 
under the restriction that we keep the cardinality of the user community fixed. 
In particular, a linear space is the only type of combinatorial configuration in which, 
for all points $p$, the point set satisfies $P=N(p)\cup\{p\}=CN(p).$ 


We want to modify the protocol so that the use of linear spaces resists the attack described in the previous section. 
A first approach is to let the user $p$ forward also his own queries. In this way he will forward the queries from $CN(p)$. 
However, this implies that the users will forward to the server more of their own queries than queries of other users. 
Indeed, if the user $p$ for every line $l$ selects a point $p'$ on $l$ with equal probability, 
then $p$ will select $p'\neq p$ with probability $\frac{1}{rk}$, and himself with probability $\frac{1}{k}>\frac{1}{rk}$. 
The server can therefore infer the real profile of a user from his apparent profile. 
There will be partial protection of the privacy of the user in front of the server. 
But if we let the protocol run for a while in order to let the user post enough queries, 
then a user's real profile will be inferable from his apparent profile. 


A compromise between these two extremes is to let the user adjust the proportion of self-submission of queries so that his real profile results uniformly distributed over the apparent profiles of the users in $CN(p)$. 
This is the strategy employed by the P2P UPIR 2 protocol, as will be illustrated below. 

\begin{definition}
Let $p_0$ be a user in a P2P UPIR community. 
We say that $p_0$'s real query profile is uniformly and independently distributed over the apparent query profiles of a set of users $A$, if, for all queries $Q\in RP(p_0)$ and for all users $p\in A$, the events ``$p$ forwards $Q$ to the server'', have equal probability and are mutually independent. 
\end{definition}
A user $p_0$ in a community of users who are executing the P2P UPIR 1 protocol from the P2P UPIR INIT protocol, selects the proxy for every query uniformly at random from $N(p_0)$, and the choices are independent.  
It is therefore clear that the $RP(p_0)$ is uniformly and independently distributed over the apparent query profiles of $N(p_0)$. 
We will now see that we can adjust the self-submission in P2P UPIR 2 and achieve a uniform and independent distribution of $RP(p_0)$ over the apparent query profiles of $CN(p_0)$.
\begin{proposition}
\label{prop:selfsubmission}
Let $p_0$ be a user in a P2P UPIR 2 community. 
Then $p_0$'s real query profile is uniformly and independently distributed over the apparent query profiles of $CN(p_0)$, if and only if $p_0$'s probability of query self-submission is $\frac{1}{|CN(p_0|}=\frac{1}{r(k-1)+1}$. 
\end{proposition}
\begin{proof}
The set of users who forwards $RP(p_0)$ to the server is $CN(p_0)$.
It is clear that $RP(p_0)$ is  uniformly distributed over $CN(p_0)$ if the probability for any user in $CN(p_0)$ to forward any of $p_0$'s queries is  $1/|CN(p_0)|$.
In particular this implies that for $RP(p_0)$ to be uniformly distributed over $CN(p_0)$,  $p_0$ should have self-submission probability $1/|CN(p_0)|$.
We have $|CN(p)|=r(k-1)+1$ for all $p$.

We will now see that, for $RP(p_0)$ to be uniformly distributed over $CN(p_0)$,
it is sufficient that $p_0$ has self-submission probability $1/(r(k-1)+1)$.

Suppose $p_0$ has self-submission probability $1/(r(k-1)+1)$.
Let $Q$ be a query in $RP(p_0)$.
The probability that $Q$ is posted to the community is $1-\frac{1}{r(k-1)+1}$.
The queries in $RP(p_0)$ are distributed by $p_0$ over his communication spaces following a uniform distribution,
so the probability that $Q$ is posted to the communication space $l$ is $\frac{1}{r}\left(1-\frac{1}{r(k-1)+1}\right)=\frac{k-1}{r(k-1)+1}$.
There are $k-1$ other users than $p_0$ connected to $l$, and they are selected using a uniform distribution, 
so the probability that any particular user $p\in N(p)$ will read and forward $Q$ is
$\frac{1}{k-1}\frac{k-1}{r(k-1)+1}=\frac{1}{r(k-1)+1}=\frac{1}{|CN(p_0)|},$
which equals the probability that $p_0$ forwards $Q$.

The choices of communication space and user are independent, so we conclude that for every query $Q$ that $p_0$ posts to the community of users, the events ``$p$ forwards $Q$ to the server'' have equal probability for all $p\in CN(p)$ and that the choices of $p$ are all mutually independent.  
\end{proof}

From now on, we will always assume that the P2P UPIR 2 protocol is implemented with self-submission $1/|CN(p)|$, as indicated by Proposition~\ref{prop:selfsubmission}. 

\subsection{Closed neighborhood attacks on P2P UPIR 2}
\label{sec:attack2}

The P2P UPIR 2 protocol with self-submission $1/|CN(p)|$ avoids the attack described in Section~\ref{sec:attack1}, when the configuration that is used is a linear space. 
However, in general, for other combinatorial configurations, the P2P UPIR 2 protocol also presents weaknesses in case of repeated queries. 
The real query profile of $p$ is independently and uniformly distributed over the apparent profiles of $CN(p)$. 
If $p$ repeats a rare query enough, then this query can be linked to him whenever the set $CN(p)$ can. 

Examples of combinatorial configurations with unique closed neighborhoods
are provided in Section~\ref{sec:unique_neighborhoods} and we discuss combinatorial configurations with $n$-anonymous closed neighborhoods in Section~\ref{sec:nanonymousclosedneighborhoods}. 
\subsection{Other attacks}
\label{sec:attack_altres}
Swanson and Stinson described an attack on the P2P UPIR 2 protocol that was based on the intersection of closed neighborhoods~\cite{SwansonStinson}. 
Following \cite{StBr_PAIS2011}, they also use a repeated rare query or variation of query, say $Q$. 
Instead of focusing on the closed neighborhood of the real owner of $Q$,  their concern is the closed neighborhood of the proxy.
The attack consists in intersecting the closed neighborhoods of the users who act as proxy for the query $Q$.
The result of the attack is set of users containing the anonymity set of the origin of the query. 
If this set is small, we have reidentification. 
It is clear that there is an analogous attack on the P2P UPIR 1 protocol, intersecting the neighborhoods of the proxies. 

These attacks are performed by a curious server, just as the attacks in Section~\ref{sec:attack1} and \ref{sec:attack2}. 
It is easy to see that an intersection attack can take place if and only if the configuration that is used in the protocol has unique neighborhoods or closed neighborhoods, respectively.
Therefore, the two types of attacks are essentially the same. 

As observed by Swanson and Stinson, we can also consider intersection attacks in which the adversary is a user in the community. 
In this case $m$ proxies collude in order to find the origin of a sequence of $l$ linked queries. 
Swanson and Stinson use other incidence structures than configurations, 
where two points may appear together in more than one block. 
In this case, the proxies can intersect the blocks over which they received $Q$. 
If we use combinatorial configurations, this does not occur, so this is a strong reason for using combinatorial configurations in P2P UPIR.  

Below we briefly list other possible attacks:
\begin{itemize}
\item The adversary can reveal the underlying combinatorial structure, by introducing users owned by him in the community. This attack was briefly discussed in \cite{StBr_CRM};
\item The adversary can determine who is in the community, since these users will have very similar apparent query profiles, and this profile will differ from apparent query profiles of users outside the community.  
\end{itemize} 
\subsection{Discussion}

We want to point out that although P2P UPIR 2 allows for the use of linear spaces without risk for intersection attacks, 
and the linear spaces have neighborhoods of maximal cardinality, the original P2P UPIR 1 protocol is still slightly simpler in implementation. 
This would be even more so, if the self-submission was expressed as a proportion of the query profile. 
Because of its simplicity, the use of the P2P UPIR 1 protocol is still justified, 
if anonymity can be ensured.



\section{On the privacy provided by  P2P UPIR}
\label{sec:privP2PUPIR}
In this section we specify the type of privacy that can be attained using P2P UPIR. 
We also show which combinatorial configurations to use in order to attain this privacy.  
\subsection{n-Confusion for P2P UPIR}
We will use the notations on database disclosure control, introduced in Section~\ref{sec:notations}.
As commented there, there is no interest in preserving the utility of the database in the transformation. 
We are only interested in minimizing the disclosure risk. 
The best result would therefore be a protected database completely free from information. 
The P2P UPIR protocol can not achieve this, as single queries contain information and are indivisible.

The purpose of P2P UPIR is to cause confusion on who is the real sender of the query. 
It is useful to have a measure of the provided confusion. 
\begin{definition}
If the cardinality of the anonymity set for the owner of any sequence of linked queries (or query) is at least $n$, 
then we say that we have \emph{$n$-confusion.} 
In this case we say that we have \emph{$n$-confusing} P2P UPIR. 
\end{definition}

If it is known who is in the community of users $P$, 
then the confusion on who is the sender of a query (i.e. the cardinality of the anonymity set), 
cannot be larger than the set of users. 
In Section~\ref{sec:attack_altres} we saw that the server can see who is in the community, 
since these users will have similar apparent query profiles. 
Therefore, the best we can aspire for is a confusion of magnitude $n=|P|$ on who is the owner of a query. 
In general, we want to cause confusion on who is the owner of a sequence of queries, also when the sequence is linkable by content.
Also in this case, the obvious upper bound for the confusion is $n\leq |P|$. 

We are interested in achieving $n$-confusion also for $n<|P|$, 
if this can be justified by other advantages, for example, as in this article, if it permits us to use the simpler P2P UPIR 1 protocol, instead of the slightly more complicated P2P UPIR 2 protocol.

In the following example we see that $n$-confusion with $n>1$ can fail to be achieved by P2P UPIR, if the sequence of linked queries contains a quasi-identifier. 
Therefore, in this case, the owner of the query sequence will not be anonymous. 
\begin{example}
\label{ex:nonanonymity}
Consider a sequence $s$ of queries posted by a user $p$ that is linkable by content. 
Observe that this does not imply that $s$ is linkable to $p$. 
Suppose that the content of the queries in $s$ gives information for linking $s$ to $p$. 
Then the anonymity set of $s$ has cardinality one, so P2P UPIR cannot provide anonymity for $p$ with respect to $s$. 

\end{example}
In the following we will always assume that $AP_t(P)$ does not contain sequences of linkable queries with quasi-identifiers.  

Under this assumption, the privacy provided by P2P UPIR in case of sequences of linked queries is anonymity; we can link the queries but we cannot link them to their owner. 
The presence of sequences of queries that are linkable because of their content, obstructs unlinkability in P2P UPIR. 
However, if we assume that the adversary cannot use the query content for the analysis, 
then unlinkability can be provided for the queries. 
This situation may occur, for example, if the range of possible queries is small. 
More precisely, it occurs if there are no rare repeated queries.  
 
Either we have anonymity or we do not. 
Anonymity is therefore provided by P2P UPIR if the protocol satisfies $n$-confusion with $n>1.$
The anonymity of a user can be broken by a collusion of the $n-1$ other users in his anonymity set. 
Therefore it is interesting to maximize $n$.  

In the previous discussion we have always assumed that the information that is available to the adversary is the same information that is available to the server. 
We have anonymity also with respect to other users if the identity of the owner of a query, or a sequence query, is known only to this user.   
This occurs if traffic analysis is prohibited, 
there are no linked sequence of queries with quasi-identifier, 
and the number of users $k$ on every communication space is large. 
If $k=2$ then one of the two users know with certainty who is the owner of the query. 
In general, a collusion of $k-1$ users is needed to deduce the identity of the query owner. 
Therefore it is interesting to maximize $k$. 
 
On the other hand, if we assume that the adversary user can see the identity of the query owner, using for example traffic analysis, or if there are linked sequences with quasi-identifier, 
then it is interesting to break up the real query profiles in small parts, in order to provide some confidentiality. 
In this case it is therefore interesting to maximize $r$, the number of communication spaces per user. 

\subsection{n-Anonymity for P2P UPIR}

It is clear that the use of P2P UPIR does not imply that the resulting database $AP_t(P)$ is $n$-anonymous in the sense of Definition~\ref{1:def:nan}.
Indeed, $AP_t(P)$ will not in general have $n$ occurrences in $n$ different records of any sensitive sequence $s$. 

The attacks in Sections~\ref{sec:attack1} and \ref{sec:attack2}  suggest that there is a quasi-identifier present in $AP_t(P)$. 
For P2P UPIR 1 and P2P UPIR 2 this quasi-identifier is the set of neighborhoods and the set of closed neighborhoods, respectively.  

Formally, before transforming the database $RP_t(P)$ using the P2P UPIR protocol transformation $\rho$, 
we first add the attribute $N(p)$ (or $CN(p)$) of $p\in P$ to $RP_t(P)$. 
The attribute $N(P)$ (resp. $CN(P)$) is invariant for the action of $\rho$, 
which in particular means that $\rho$ preserves its quasi-identifying property.  
According to Definition \ref{1:def:nan}, 
in order to make $AP_t(P)$ $n$-anonymous with respect to this quasi-identifier, 
we have to ensure that every element of the set of neighborhoods (resp. the set of closed neighborhoods)  occurs at least $n$ times in $AP_t(P)$. 

\begin{definition}\label{def:nanconf1}
We say that a combinatorial configuration has $n$-\emph{anonymous neighborhoods} (resp. $n$-\emph{anonymous closed neighborhoods}), 
if every neighborhood (resp. closed neighborhood) of a point, 
is the neighborhood (resp. closed neighborhood) of at least $n$ points. 
\end{definition}

\begin{figure}

\begin{tabular}{ccc}
\resizebox{!}{0.22\textwidth}{\includegraphics{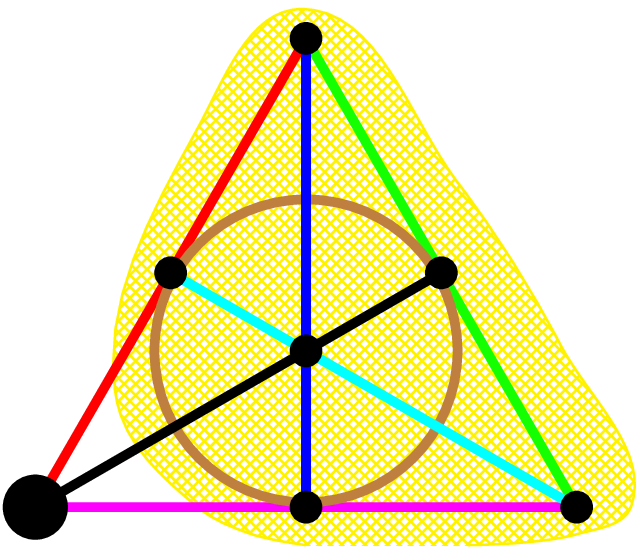}}
&
\resizebox{!}{0.22\textwidth}{\includegraphics{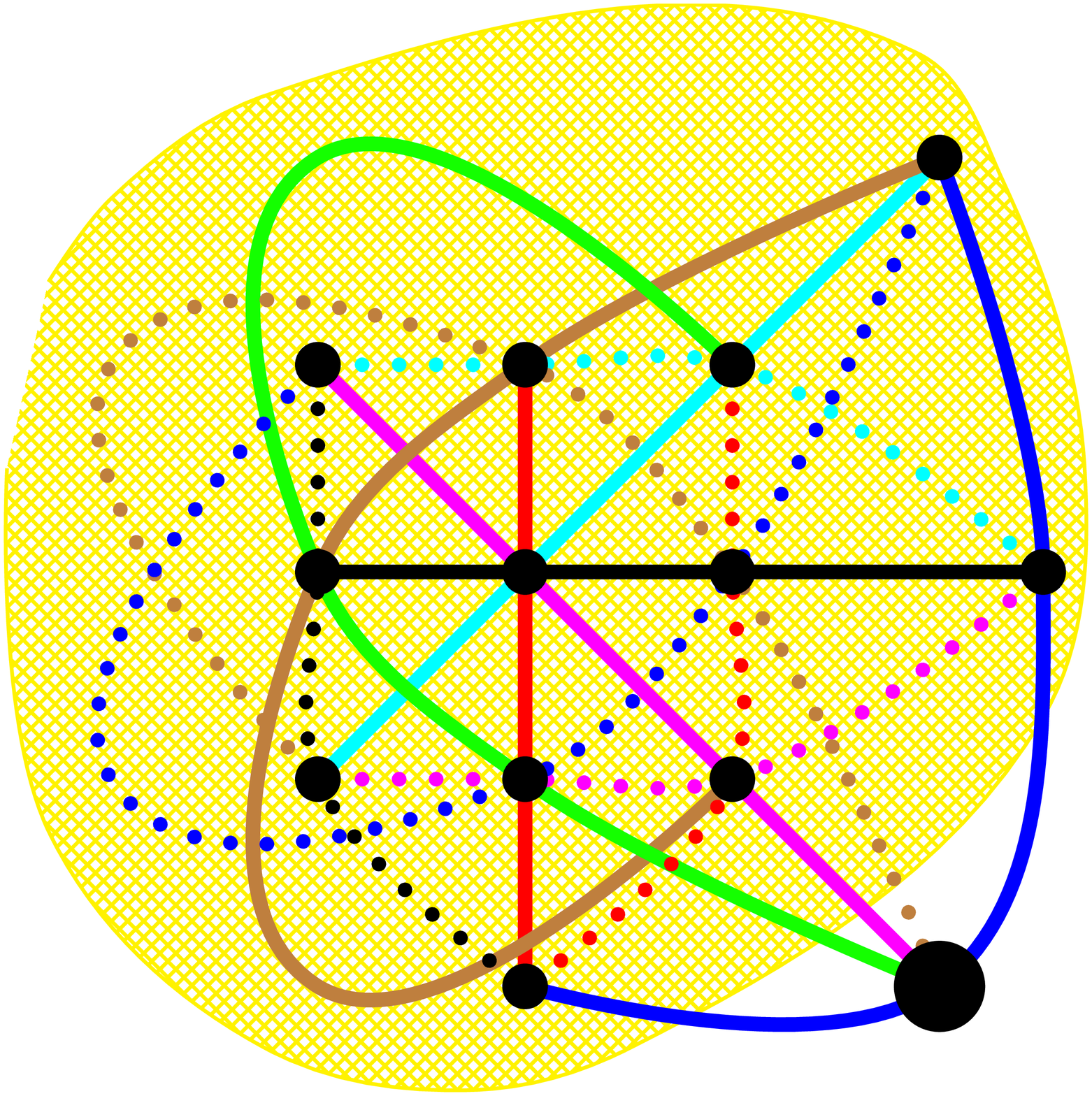}}
&
\resizebox{!}{0.22\textwidth}{\includegraphics{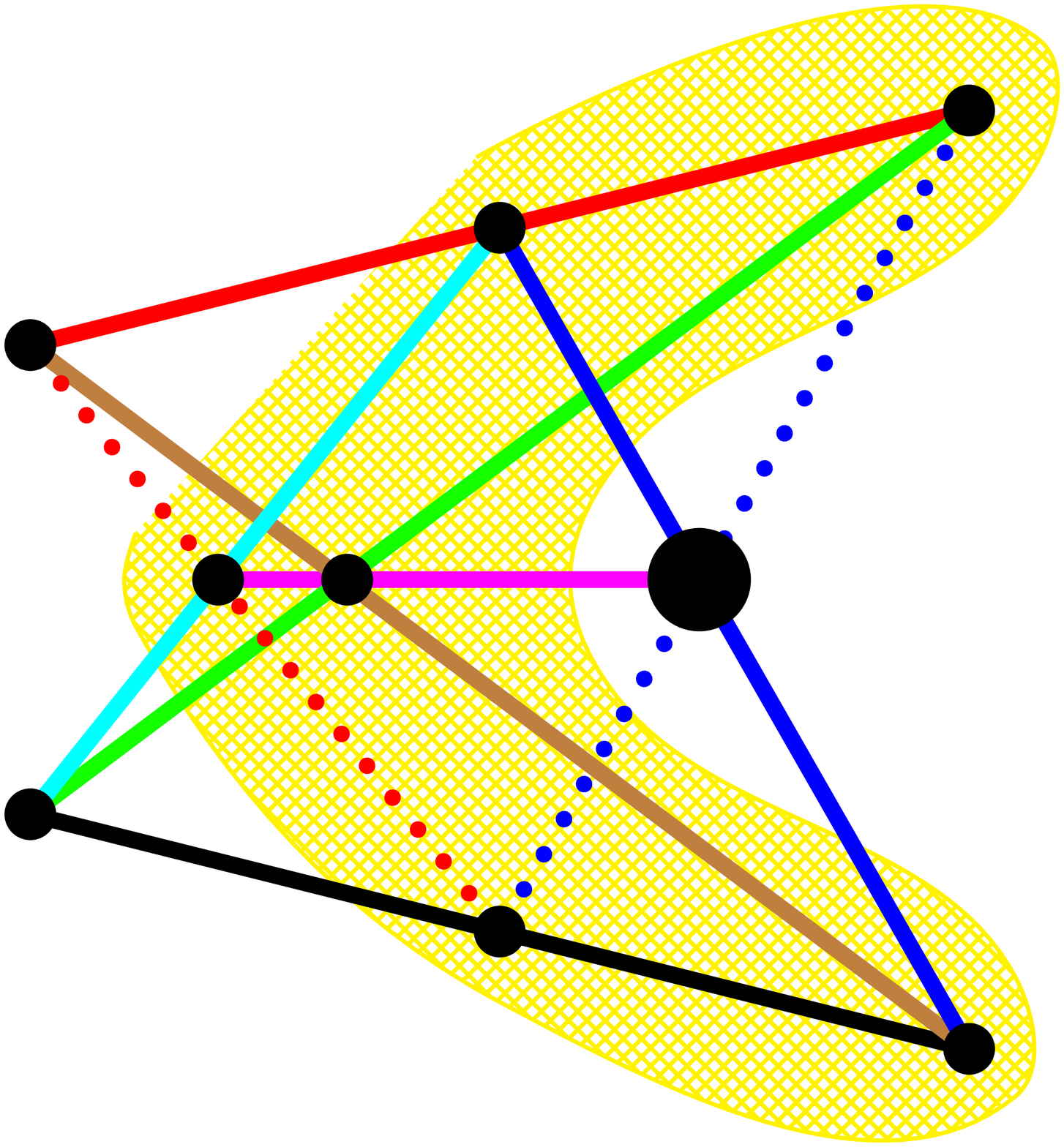}}
\\
The neighborhood
&
The neighborhood 
&
The neighborhood\\ 
of a point in & of a point in & of 3 points in \\
$\mathbb{P}(\mathbb{F}_2)$, the Fano plane& $\mathbb{P}(\mathbb{F}_3)$& the Pappus configuration
\end{tabular}
\end{figure}

\subsection{Achieving n-confusion using n-anonymous neighborhoods and n-anonymous closed neighborhoods}
\label{sec:nconfnan}
In this section we show that the P2P UPIR protocol can offer $n$-confusion, if we use the correct type of combinatorial configuration. 
\begin{proposition}
\label{prop:p2pnconf}
Under the assumption that queries can be linked by content, but there are no sequences of linked queries containing quasi-identifiers, the P2P UPIR 1 protocol (resp. the P2P UPIR 2 protocol) provides $n$-confusion if it is implemented with a combinatorial configuration with $n$-anonymous neighborhoods (resp. $n$-anonymous closed neighborhoods).

However, if the sequence of queries cannot be linked by content, then we have $n$-confusion with $n$ equal to the cardinality of the neighborhoods and the closed neighborhoods, respectively. 
\end{proposition}

\begin{proof}
We will prove the result for P2P UPIR 1. The proof for P2P UPIR 2 is analogous. 
 
Following the notation in Section~\ref{sec:disclosurecontrolp2pupir}, we apply the P2P UPIR 1 protocol $\rho$ to the database $RP_t(P)$ and obtain $AP_t(P)$.  
Suppose that the adversary is allowed to analyze the content of the queries and is able to correctly link a sequence $s$ of queries to each other, as having the same origin, say, the user $p_0$.    
We have assumed that $s$ does not contain a quasi-identifier, which could identify $p_0$ as the origin of $s$, by content alone. 

Because of the properties of $\rho$, all queries in $s$ will be in the records of $AP_t(P)$ that correspond to $N(p_0)$. 
The anonymity set of $s$ is the intersection of the neighborhoods of the neighbors of user $p_0$, that is, $\bigcap_{p\in N(p_0)}N(p).$  
If the combinatorial configuration has $n$-anonymous neighborhoods, then this intersection has cardinality at least $n$, so we have $n$-confusion. 

If query sequences cannot be linked by content, 
then the anonymity set for the owner of a single query is the neighborhood of the proxy of the query. 
This anonymity set has cardinality $r(k-1)$, so in this case we have $n$-confusion with $n=r(k-1)$. 
\end{proof}

\section{Combinatorial configurations with unique neighborhoods or unique closed neighborhoods}
\label{sec:unique_neighborhoods}
In this section we give examples of combinatorial configurations that have unique neighborhoods or unique closed neighborhoods for all points. 
We saw in Section~\ref{sec:compprot} that such configurations should be avoided for the use in P2P UPIR~1 and P2P UPIR~2 respectively. 

We provide examples of combinatorial configurations with 
\begin{itemize}
\item anonymous neighborhoods but unique closed neighborhoods (as the combinatorial configurations with deficiency one),
\item unique neighborhoods but anonymous closed neighborhoods (as the linear spaces), and 
\item unique neighborhoods and unique closed neighborhoods (as the pentagonal geometries without opposite line pairs). 
\end{itemize}

\subsection{Examples of combinatorial configurations with unique neighborhoods}

In this section we give examples of combinatorial configurations with unique neighborhoods.
\begin{proposition}
The linear spaces have unique neighborhoods. 
\end{proposition}
\begin{proof}
In a regular linear space every pair of points is collinear.  
Therefore, for any point $p$, the neighborhood $N(p)$ is all the point set except for $p$, so that $p$ is the only point with neighborhood $N(p)$. 
\end{proof}

A \emph{triangle} in a combinatorial configuration is a set of three distinct points such that they are pairwise collinear on three distinct lines. 
A combinatorial configuration is \emph{triangle-free} if it has no triangles. 
\begin{proposition}\label{2.1:thm:id_conf}
A triangle-free combinatorial configuration, not a graph, has unique neighborhoods.   
\end{proposition}
\begin{proof}
Let $C=(P,L,I)$ be a triangle-free $(r,k)$-configuration with $k>2$ (so that it is not a graph). 
Fix a point $p_0\in P$ and let $p_1, p_2\in N(p_0)$ be two points collinear with $p_0$. 
Let $p_3\in P$ be a point such that $N(p_0)=N(p_3)$. 
Then $p_3$ is collinear with $p_1$ and $p_2$, but not with $p_0$, 
so that there is no line through all the four points $p_0$, $p_1$, $p_2$ and $p_3$. 
Therefore, $p_1$ and $p_2$ can not be collinear, because if they were, then at least one of the triples $p_0$, $p_1$, $p_2$ or $p_1$, $p_2$, $p_3$ would form a triangle. 
In other words, no pair of points in $N(p_0)=N(p_3)$ is collinear. 
Therefore the number of points on every line in $C$ is $k=2$, because if $k>2$, then there would be at least one pair of collinear points $p,q\in N(p_0)=N(p_3)$.  
We deduce that, whenever $k>2$,  given a point $p_0\in P$ there is no point $p_3\in P$ distinct from $p_0$ such that $N(p_0)=N(p_3)$. 
\end{proof}

A \emph{pentagonal geometry} is a combinatorial configuration in which, for any point $p$, all points not in the closed neighborhood of $p$, are on the same line~\cite{BaBaDeSt}. 
This line is called the \emph{opposite line} $p^{opp}$ of $p$.
\begin{proposition}
A pentagonal geometry has unique neighborhoods. 
\end{proposition}
\begin{proof}
Let $p\neq q$ be two points in the pentagonal geometry. 
Suppose that $q$ is not on $p^{opp}$. Then $p$ and $q$ are collinear, so $q\in N(p)$. But $q\not\in N(q)$, and we deduce that $N(p)\neq N(q)$.
Now suppose that $q$ is on $p^{opp}$. We have that $p$ is not on $p^{opp}$, so $p^{opp}\neq q^{opp}$. 
Let $x\neq q$ be a point on $p^{opp}$. Then $x\in N(q)$, but $x\not\in N(p)$, so $N(p)\neq N(q)$.   
\end{proof}

\subsection{Combinatorial configurations with unique closed neighborhoods}
In this section we give examples of combinatorial configurations with unique closed neighborhoods. 

The parameters of any combinatorial $(v,b,r,k)$-configuration satisfies the inequality $v\geq r(k-1)+1$. 
We have the equality $v=r(k-1)+1$ if and only if we have a linear space.  
The \emph{deficiency} of a combinatorial $(v,b,r,k)$-configuration is the number $v-(r(k-1)+1)$. 
\begin{proposition}
\label{prop:def1}
A combinatorial configuration with deficiency one has unique closed neighborhoods. 
\end{proposition}
\begin{proof}
In a combinatorial configuration with deficiency one, for any point $p$ there is only one point in the complement of $CN(p)$, the \emph{anti-podal point} of $p$. 
The anti-podal points come in pairs, so every point has a unique anti-podal point and therefore, a unique closed neighborhood. 
\end{proof}
From the proof of Proposition~\ref{prop:def1} we also deduce that combinatorial configurations with deficiency one have 2-anonymous neighborhoods; any point and its anti-podal point share neighborhood. 

It can be proved that if two points $p$ and $q$ in a pentagonal geometry share the same opposite line $l$, 
then all points in $l$ will have the same opposite line: the line spanned by $p$ and $q$. 
Such a pair of lines is called an opposite line pair.  
\begin{proposition}
A pentagonal geometry with no pair of opposite lines has unique closed neighborhood. 
\end{proposition}
\begin{proof}
For any point $p$ in the pentagonal geometry, the set of points on $p^{opp}$ is the complement of $CN(p)$. 
If the pentagonal geometry has no opposite line pair, then all points have unique opposite lines, hence unique closed neighborhoods.
\end{proof}

\section{Combinatorial configurations with n-anonymous neighborhoods}
\label{sec:combconfnanneighbors}
In Section~\ref{sec:nconfnan} we saw that combinatorial configurations with $n$-anonymous neighborhoods are interesting for use with P2P UPIR 1. 
\subsection{Examples of combinatorial configurations with \\n-anonymous neighborhoods}
Here we give an important example of a family of combinatorial configurations with $n$-anonymous neighborhoods. 
\begin{proposition}\label{2.3:thm:transnan}
A transversal design $TD(k,n)$ has $n$-anonymous neighborhoods.
\end{proposition}
\begin{proof}
The point set of the transversal design can be partitioned into $k$ groups of cardinality $n$, such that the points in the same group are not collinear. 
Any pair of points from different groups is contained in exactly one line. 
This implies that the $n$ points in the same group all have the same neighborhood. 
\end{proof}
The transversal design $TD(k,n)$ in this construction is a combinatorial $(nk,n^2,n,k)$-configuration. 
Hence the construction provides a combinatorial configuration with $n$-anonymous neighborhoods that is suitable for $nk$ P2P UPIR users and requires the use of $n^2$ communication spaces. 
 

As we saw in \ref{sec:transversaldesigns}, transversal designs can be constructed using latin squares and many transversal designs can be easily constructed using affine planes. 
A transversal design constructed from an affine plane of order $q$  has parameters $(q^2,q^2,q,q)$. 
The use of the affine plane of order $2$ gives an ordinary square with 4 points and 4 lines with 2 points on every line. 
The use of the affine plane of order 3 gives the Pappus configuration.
\subsection{A characterization of the combinatorial configurations with n-anonymous neighborhoods}
We will now characterize the combinatorial configurations with $n$-anonymous neighborhoods exactly. 
\begin{proposition}\label{thm:charnanonym}
A combinatorial $(v,b,r,k)$-configuration with $n$-anonymous neighborhoods is a combinatorial configuration that satisfies the following conditions:
\begin{itemize}
\item There exists a partition $G=\{g_i\}_{i=1}^{m}$ of the point set such that the points in the same part are not collinear, and  $|g_i|\geq n$ for all $i\in [1,\dots,m]$;
\item We have that $r\geq n$ and $m\geq k$. 
\end{itemize}
\end{proposition}

\begin{proof}

Let $C=(P,L,I)$ be a combinatorial configuration with $n$-anonymous neighborhoods. 
Then every point $p\in P$ shares its neighborhood $N(p)$ with $n-1$ other points. 
``Having the same neighborhood'' is a binary relation which is obviously
\begin{itemize}
\item reflexive ($p$ has the same neighborhood as $p$);
\item symmetric (if $N(p)=N(q)$ then $N(q)=N(p)$);
\item transitive (if $N(p_1)=N(p_2)$ and $N(p_2)=N(p_3)$, then \\$N(p_1)=N(p_3)$). 
\end{itemize}
So it is an equivalence relation and defines a partition $G=\{g_1,\dots,g_m\}$ of the point set, in which $|g_i|\geq n$ for all $g_i\in G$. 
We will call the parts $g_i\in G$ groups.  
The neighborhood $N(p)$ of the point $p$ is defined as the set of points that are collinear with $p$, and different from $p$. 
In particular, if two points $p$ and $q$ satisfy $N(p)=N(q)$, 
then they are not collinear, since if they were, then $p\in N(q)$ which would imply $p\in N(p)$. 
Therefore points in the same group are not collinear.

For the bound on $r$, consider a pair of collinear points $p$ and $q$. 
Let $g$ be the group containing $p$. 
All points in $g$ have the same neighborhood, so $q\in N(p')$ for every $p'\in g$. 
No line contains two points in $g$, and we deduce that there are at least $|g|\geq n$ lines through $q$, so that $r\geq n$. 

Regarding the number of points on every line $k$, we see that, since points in the same group are not collinear, it is clear that any line contains $k$ distinct points from $k$ distinct parts of $G$, so that $k\leq m$.  
\end{proof}

There are indeed, $n$-anonymous combinatorial configurations which are not transversal designs. 
\begin{example}
\label{ex:3an}
Consider the combinatorial $(36,72,6,3)$-configuration with point set $P=\{1,\dots,36\}$ and line set as in Table~\ref{table:conf}.  

\begin{table}
\caption{Line set of a the combinatorial $(36,72,6,3)$ configuration in Example~\ref{ex:3an}.}
\label{table:conf}
\begin{center}
\begin{tabular}{llll}
$\begin{array}{c}
\{\{1,4,7\},\\
\{1,5,8\},\\
\{1,6,9\},\\
\{2,4,8\},\\
\{2,5,9\},\\
\{2,6,7\},\\
\{3,4,9\},\\
\{3,5,7\},\\
\{3,6,8\},\\
\{1,10,13\},\\
\{1,11,14\},\\
\{1,12,15\},\\
\{2,10,14\},\\
\{2,11,15\},\\
\{2,12,13\},\\
\{3,10,15\},\\
\{3,11,13\},\\
\{3,12,14\},\end{array}$
&
$\begin{array}{c}
\{4,16,19\},\\
\{4,17,20\},\\
\{4,18,21\},\\
\{5,16,20\},\\
\{5,17,21\},\\
\{5,18,19\},\\
\{6,16,21\},\\
\{6,17,19\},\\
\{6,18,20\},\\
\{7,22,25\},\\
\{7,23,26\},\\
\{7,24,27\},\\
\{8,22,26\},\\
\{8,23,27\},\\
\{8,24,25\},\\
\{9,22,27\},\\
\{9,23,25\},\\
\{9,24,26\},\end{array}$
&
$\begin{array}{c}
\{10,28,31\},\\
\{10,29,32\},\\
\{10,30,33\},\\
\{11,28,32\},\\
\{11,29,33\},\\
\{11,30,31\},\\
\{12,28,33\},\\
\{12,29,31\},\\
\{12,30,32\},\\
\{13,16,34\},\\
\{13,17,35\},\\
\{13,18,36\},\\
\{14,16,35\},\\
\{14,17,36\},\\
\{14,18,34\},\\
\{15,16,36\},\\
\{15,17,34\},\\
\{15,18,35\},\end{array}$
&
$\begin{array}{c}
\{19,22,31\},\\
\{19,23,32\},\\
\{19,24,33\},\\
\{20,22,32\},\\
\{20,23,33\},\\
\{20,24,31\},\\
\{21,22,33\},\\
\{21,23,31\},\\
\{21,24,32\},\\
\{25,28,34\},\\
\{25,29,35\},\\
\{25,30,36\},\\
\{26,28,35\},\\
\{26,29,36\},\\
\{26,30,34\},\\
\{27,28,36\},\\
\{27,29,34\},\\
\{27,30,35\}\}\end{array}$
\end{tabular}
\end{center}
\end{table}
It is clear that this combinatorial $(36,72,6,3)$-configuration is $3$-anonymous, but $k=3<12=m$ and $r=6>3=n$. 
We also observe that $rk=18$ divides $v=36$ and $b=72$.
The groups in the partition are given by Table~\ref{table:groups}.
\begin{table}
\caption{The partition of the point set into anonymity sets of the combinatorial$(36,72,6,3)$-configuration with 3-anonymous neighborhoods in Example~\ref{ex:3an}.}
\label{table:groups}
 \begin{center}
\begin{tabular}{ccccc}
$\{\{1,2,3\},$&$\{4,5,6\},$&$\{7,8,9\},$&$\{10,11,12\},$\\
$\{13,14,15\},$&$\{16,17,18\},$&$\{19,20,21\},$&$\{22,23,24\},$\\
$\{25,26,27\},$&$\{28,29,30\},$&$\{31,32,33\},$&$\{34,35,36\}\}$
\end{tabular}\end{center}
\end{table}
 
\end{example}


\subsection{Optimal configurations for $n$-anonymous P2P UPIR 1}

The privacy provided to the users of $n$-anonymous P2P UPIR 1 is $n$-confusion, where $n$ is the cardinality of the anonymity sets. 
The points in the same anonymity set have the same neighborhood. 
As we saw in Proposition~\ref{prop:p2pnconf}, if query sequences cannot be linked by content, then we have $r(k-1)$-confusion, since this is the cardinality of the neighborhood of the proxy of a single query. 
It is therefore interesting to maximize both $n$ and $r(k-1)$, although which one is the most important may depend on the context. 
 
In a combinatorial configuration with $n$-anonymous neighborhoods, the anonymity set and the neighborhood of a point are disjoint. 
Therefore, maximizing $n$ and $r(k-1)$ simultaneously is the same as requiring $v=n+r(k-1)$, so that the anonymity set and the neighborhood together form the entire point set of the configuration. 
It is easy to see that a transversal design satisfies this condition, 
since a point $p$ in a transversal design is neighbor with all points that are not in the anonymity set of $p$.  

Indeed, if we in Lemma~\ref{thm:charnanonym} 
add a restriction on regularity of the group cardinalities, and maximize $n$, what we get are exactly the transversal designs. 
\begin{theorem}
In a combinatorial $(v,b,r,k)$-configuration $C$ with $n$-anonymous neighborhoods and anonymity partition $G=\{g_i\}_{i=1}^m$ and $|g_i|=n$ for all $i\in [1,\dots,m]$, 
we have that
\begin{center}$r=n$ if and only if $m=k$.\end{center} 
In this case $C$ is a transversal design $TD(k,n)$ and $v=kn$, $b=n^2$. 
\end{theorem}
\begin{proof}


Since the configuration is connected, if $r=n$, then necessarily $k=m$. 
On the other hand, if $k=m$, then necessarily $r=n$, since if we fix one part $g\in G$ and a point $p\in g$, 
then a line through $g$ has $k$ points through $k=m$ distinct parts $g\in G$, so the line have one point in every part in $G$. 
For any part $g'\in G$ different from $g$ there is also a total of $n$ lines through $p$. 
Since these lines have one point in every part of $G$, we get  $r=n$. 

A transversal design is a uniform group divisible design in which the number of groups $|G|$ equals the length of the blocks $k$. 
We have seen that an $n$-anonymous combinatorial $(v,b,n,m)$-configuration such that $|g_i|=n$ and $m=k$ satisfy exactly these conditions, 
so it is a transversal design $TD(k,n)$.
\end{proof}


\section{Combinatorial configurations with n-anonymous closed neighborhoods}
\label{sec:nanonymousclosedneighborhoods}
In Section~\ref{sec:nconfnan} we saw that combinatorial configurations with $n$-anonymous closed neighborhoods are interesting for P2P UPIR 2. 
By now, the reader is already familiar with the most important example of combinatorial configurations with $n$-anonymous closed neighborhoods. 
\begin{proposition}
\label{prop:lin_anon}
A linear space on $n$ points has $n$-anonymous closed neighborhoods.
\end{proposition}
\begin{proof}
In a linear space all points are collinear.
\end{proof}

\subsection{Combinatorial configurations with n-anonymous closed neighborhoods from combinatorial configurations with n-anonymous neighborhoods}
\label{sec:suboptimalnanonym}
After all, there is not much difference between the definition of $n$-anonymous neighborhoods and the definition of $n$-anonymous closed neighborhoods. 
The next Theorem~\ref{thm:nanonym(III)from(I)} shows that we can use combinatorial configurations with the former property to construct combinatorial configurations with the latter property. 
\begin{theorem}\label{thm:nanonym(III)from(I)}
Let $C$ be a combinatorial $(v,b,r,k)$-configuration with $n$-anonymous neighborhoods such that $k|n$. 
Then there also exists a combinatorial $(v,b+n,r+1,k)$-configuration $C'$ with $n$-anonymous closed neighborhoods.
\end{theorem}
\begin{proof}
Let $C$ be as stated above. Then every point shares neighborhood with exactly $n$ more points. 
Theorem~\ref{thm:charnanonym} implies that in $C$ there is a partition $G$ of the point set so that points in the same partition are the points with the same neighborhood. 
This implies that points in the same partition are not collinear. 
Define $C'$ by adding $k\frac{n}{k}=n$ new lines, 
so that every new line contains only points from the same part of $G$. 
Let $A$ be a set of points with the same neighborhood in $C$. 
For any $p\in A$ there are $k-1$ other points $p_1,\ldots,p_{k-1}$ in $A$ 
collinear with $p$ by one of the new lines, such that $CN(p)=CN(p_i)$ for $i=1,\ldots,k-1$. 
This concludes the proof.
\end{proof}
As a corollary of Theorem~\ref{thm:nanonym(III)from(I)} we get that an affine plane of order $k$ is a combinatorial configuration with $k$-anonymous closed neighborhoods.  
Just apply the construction in the proof of Theorem~\ref{thm:nanonym(III)from(I)} to  a transversal design $TD(k,k)$. 

But an affine plane of order $k$ is a linear space on $v=k^2$ points, 
so we already know from Proposition~\ref{prop:lin_anon} that it is a $k^2$-anonymous combinatorial configuration for P2P UPIR 2. 
Indeed, $n$-anonymity implies $m$-anonymity for all $m\leq n$. 
Observe though that in general the combinatorial $(r,k)$-configuration constructed in Theorem~\ref{thm:nanonym(III)from(I)} is $k$-anonymous but not $m$-anonymous for $m>k$. 

Not all combinatorial configurations with $n$-anonymous closed neighborhoods can be obtained using the construction in Theorem~\ref{thm:nanonym(III)from(I)}. 
For example, we can not use this method to construct a finite projective plane. 

\subsection{Optimal configurations for P2P UPIR 2}

The P2P UPIR 2 was designed to provide $n$-confusion with linear spaces. We have the following result. 
\begin{theorem}\label{thm:optpp}
A regular linear space on $v$ points provides $n$-confusing P2P UPIR 2 with $n=v$. 
This is optimal. 
\end{theorem}
\begin{proof}
It is immediate that a linear space provides $n$-confusing P2P UPIR 2 with $n=v$. 
Since the confusion can not be larger than the total number of points $v$ in the configuration, this is optimal. 
%
%
\end{proof}

A linear space is a $(v,k,1)$-BIBD. More generally, in a $(v,k,\lambda)$-BIBD any two points are connected by $\lambda\geq 1$ lines, so also in this case we have optimal $n$-confusing P2P UPIR 2. 
However, then the BIBD is not a combinatorial configuration. 
In this article we have provided reasons that justify the use of combinatorial configurations for P2P UPIR. 
As was observed in \cite{SwansonStinson}, other incidence structures are also interesting, in particular if it is assumed that colluding users can communicate over channels that are exterior to the protocol.

\section{Conclusions}
We have presented two different P2P UPIR protocols, P2P UPIR 1 in which the users do not self-submit and P2P UPIR 2 in which they do. 
Then we described an attack on P2P UPIR 1, based on unique neighborhoods, 
and adjusted the self-submission for P2P UPIR 2 in order to avoid neighborhood attacks on linear spaces. 
We also showed that P2P UPIR 2 is still vulnerable to closed neighborhood attacks, if the closed neighborhoods in the combinatorial configurations are unique. 
We gave examples of combinatorial configurations with unique neighborhoods and unique closed neighborhoods.  

Then we presented the combinatorial configurations with $k$-anonymous neighborhoods and $k$-anonymous closed neighborhoods, respectively. 
We characterized, as $n$-confusion, the privacy provided by a P2P UPIR protocol that uses one of the combinatorial configurations from these families. 
Finally we studied the combinatorial configurations with $k$-anonymous neighborhoods and $k$-anonymous closed neighborhoods. 
We distinguished the transversal designs and the linear spaces as optimal configurations for P2P UPIR from these two families, respectively. 

We want to point out that there are two trivial ways to connect communication spaces and users; the all-to-all and the one-to-all distributions. 
As combinatorial structures both can be interpreted as degenerated linear spaces, since every pair of points (users) share exactly one line (communication space). 
Used in the P2P UPIR 2 protocol, they provide the same anonymity in front of the server as does a linear space with the same number of points. 
However, with respect to other users, there is no anonymity in the all-to-all distribution and confidentiality is lost in the all-to-one distribution.  
Consequently, the reason why non-degenerated combinatorial configurations are interesting for P2P UPIR 2, is because they offer some anonymity and confidentiality with respect to the other users in the community.

\section*{Acknowledgements}
The authors want to thank Maria Bras-Amor\'os, Douglas R. Stinson, Colleen Swanson, Vicen\c{c} Torra and the anonymous referees for useful discussions and suggestions. 
Partial support by the Spanish MEC projects ARES (CONSOLIDER ~INGENIO ~2010 ~~CSD2007-00004) and RIPUP (TIN2009-11689) is acknowledged. The authors are with the UNESCO Chair in Data Privacy, but their views do not necessarily reflect those of UNESCO nor commit that organization. 

\bibliographystyle{plain}

\end{document}